\documentclass[letterpaper, 10 pt, conference]{ieeeconf}   
\IEEEoverridecommandlockouts                  
\usepackage{cite}
\usepackage{amsmath}
\usepackage{circuitikz}
\usepackage{algorithm2e}
\usepackage{array}
\usepackage{interval}
\usepackage{amssymb,amsfonts,mathrsfs,mathtools}
\usepackage{color}
\usepackage{booktabs,makecell}
\usepackage{pict2e}
\usepackage{optidef}
\SetKwComment{Comment}{/* }{ */}
\usepackage{verbatim}
\usetikzlibrary{calc}
 \newcommand\jw{j\omega}

\newcommand\dist[2]{\mathrm{dist}\rbkt{#1, #2}}

\newcommand\sbkt[1]{\left[#1\right]}
\newcommand\rbkt[1]{\left(#1\right)}
\newcommand\ininf[2]{\langle #1\,, #2 \rangle}

\newcommand\rlinf{\mathcal{RL}_{\infty}}

\newcommand\rhinf{\mathcal{RH}_{\infty}}

\newcommand\ccp{\bar{\mathbb{C}}_+}
\newcommand\cop{{\mathbb{C}}_+}

\newcommand\cm{{\mathbb{C}}^m}
\newcommand\mm{{m\times m}}

\newcommand\cmm{{\mathbb{C}}^{m\times m}}

\newcommand\sg{\mathrm{SG}}

\newcommand\abs[1]{\left|#1\right|}
\newcommand\rep{{\rm Re}}

\newcommand{\norm}[1]{\left\lVert#1\right\rVert}

\newcommand{\tbt}[4]{\begin{bmatrix}#1&#2\\#3&#4\end{bmatrix}}

\newcommand{\tbo}[2]{\begin{bmatrix}#1\\#2\end{bmatrix}}

\newcommand{\stbt}[4]{\left[\begin{smallmatrix} #1&#2\\#3&#4\end{smallmatrix}\right]}
\newcommand{\stbo}[2]{\left[\begin{smallmatrix} #1\\#2\end{smallmatrix}\right]}

\newcommand{\bi}{\begin{itemize}}\newcommand{\ei}{\end{itemize}}
\newcommand{\be}{\begin{equation}}\newcommand{\ee}{\end{equation}}
\newcommand{\bex}{\begin{equation*}}\newcommand{\eex}{\end{equation*}}
\newcommand{\bax}{\begin{align*}}\newcommand{\eax}{\end{align*}}
\newcommand{\bc}{\begin{center}}\newcommand{\ec}{\end{center}}

\def\gof{\mbox{${P}\,\#\,C$}}
\def\goflti{\mbox{${P}\,\#\,C$}}

\newtheorem{theorem}{Theorem}
\newtheorem{proposition}{Proposition}
\newtheorem{definition}{Definition}
\newtheorem{lemma}{Lemma}
\newtheorem{remark}{Remark}
\newtheorem{example}{Example}
\newtheorem{problem}{Problem}
\newtheorem{corollary}{Corollary}

\intervalconfig{soft open fences}
\usetikzlibrary{shapes,arrows}

\makeatletter
   \let\NAT@parse\undefined
   \makeatother
\usepackage{hyperref}
\hypersetup{colorlinks = true,
            linkcolor = red,
            urlcolor  = black,
            citecolor = blue,
            anchorcolor = blue}

\title{\LARGE \bf Graphical Dominance Analysis for Linear Systems: \\A Frequency-Domain Approach}

\author{Chao Chen, Thomas~Chaffey and Rodolphe Sepulchre,~\IEEEmembership{Fellow,~IEEE}   
\thanks{This work was supported by the European Research Council under the Advanced ERC Grant Agreement SpikyControl n.~101054323.} 
\thanks{Chao Chen is with the Department of Electrical and Electronic Engineering, The University of Manchester, UK. (chao.chen@manchester.ac.uk)}
\thanks{Thomas~Chaffey is with the School of Electrical and Computer Engineering, The University of Sydney, Australia. (thomas.chaffey@sydney.edu.au)}
\thanks{Rodolphe~Sepulchre is with the STADIUS center of Department of Electrical Engineering (ESAT), KU Leuven, Leuven, Belgium and the Department of Engineering, University of Cambridge, Cambridge, UK. (rodolphe.sepulchre@kuleuven.be)}}
 
\begin{document}
\maketitle
\pagestyle{empty}

\begin{abstract}
We propose a frequency-domain approach to dominance analysis for multi-input multi-output (MIMO) linear time-invariant systems. The dominance of a MIMO system is defined to be the number of its poles in the open right half-plane. Our approach is graphical: we define a frequency-wise notion of the recently-introduced scaled graph of a MIMO system plotted in a complex plane. The scaled graph provides a bound of the eigenloci of the system, which can be viewed as a robust MIMO extension of the classical Nyquist plot. Our main results characterize sufficient conditions for quantifying the dominance of a closed-loop system based upon separation of scaled graphs of two open-loop systems in a frequency-wise manner. The results reconcile existing small gain, small phase and passivity theorems for feedback dominance analysis.
\end{abstract}

\section{Introduction}

Behavior of a high-dimensional system may be approximated by that of a low-dimensional one. Such a property is known as dominance \cite{Forni:18, Miranda:18, Padoan:19, Padoan:21, Che:23}. The dominance of a system indicates the asymptotic behavior of the system quantitatively. A special class of dominant systems is the set of asymptotically stable, or $0$-dominant, systems. We refer the reader to the seminal work \cite{Forni:18} for a comprehensive look at differential dominance analysis of nonlinear systems. 

The focal point of this paper lies in \emph{robust dominance analysis} for multi-input multi-output (MIMO) linear time-invariant (LTI) systems via a frequency-domain approach. Our study is originally motivated from  intrinsic local robustness of complicated and rich behaviors in nonlinear biological systems and neural networks beyond local stability. A linearized model of such a network often admits unstable poles. It is desirable to \emph{preserve the local behavior} subject to uncertainties, see, for example, \cite{Hara:22}. A MIMO system $G\in \rlinf^\mm$ is called $p$-dominant if it has $p$ poles in the open right half-plane. The property of $p$-dominance measures the levels of instability of the system. The main benefit that comes with $p$-dominance is that classical notions of system gain and positive realness can be adapted to generic unstable systems quantitatively. In a nutshell, a single-input single-output (SISO) $p$-dominant system $G$ is called \emph{generalized positive real} \cite{Anderson:68, Brockett:67} if $\rep[G(\jw)] \geq 0$
for all $\omega \in \interval{-\infty}{\infty}$. Similarly,  $G$ is said to have a \emph{conditional gain} ($\mathcal{L}_\infty$-norm) $\gamma$ \cite{Takeda:73} if $\gamma = \sup_{\omega\in \interval{-\infty}{\infty}}  \abs{G(\jw)} <\infty$.

\begin{example}\label{ex: intro}
  Consider a simple $1$-dominant system $\frac{1}{1-s}$. It is generalized positive real and has a conditional gain one. Such a property can be easily demonstrated by sketching the Nyquist plot $\frac{1}{1-\jw}$ over $\omega\in \interval{-\infty}{\infty}$ that is contained in the intersection of the unit disk and right half-plane.
\end{example}

As a counterpart to stability analysis, instability analysis of feedback interconnections involving nonlinear subsystems characterized by generalized positive realness or conditional gains may be traced back to the 1960s \cite{Brockett:67, Willems:69Instability, Takeda:73, Willems:73, Hill:83, Georgiou:97_gain, Khong:24}. Brockett and Lee \cite{Brockett:67} developed an instability version of the celebrated circle criterion using generalized positive realness. Willems \cite{Willems:69Instability} adopted a noncausality approach to an instability counterpart of the conicity theorem. Takeda and Bergen \cite{Takeda:73} proposed the small-gain and passivity instability theorems based on orthogonal decompositions. A dissipativity approach was proposed by Willems \cite{Willems:73} with the notion of cyclo-dissipativity. This notion was subsequently taken up by Hill and Moylan for instability analysis \cite{Hill:83}. A differential dissipativity approach was proposed by Forni and Sepulchre for dominance analysis \cite{Forni:18} as quantitative instability analysis. Recently, robust instability analysis has been conducted in \cite{Khong:24} using the $\mathcal{L}_2$-gap metric and coprime factorizations. 

Robustness of feedback dominance is of great importance in LTI systems analysis and synthesis, an has not yet been well investigated. When dealing with SISO systems, robust feedback dominance can be conveniently analyzed using the Nyquist criterion. In this case the dominance of a positive feedback loop can be inferred from the number of encirclements of the critical point ``$1$''. The distance between ``$1$'' and the Nyquist plot of the feedback loop serves as an important \emph{robustness indicator} for feedback dominance, termed the dominance margin \cite{Padoan:19}. A related notion of exact instability margin for SISO systems was proposed in \cite{Hara:24}.

For MIMO systems, however, there does not exist an obvious and widely-accepted \emph{graphical} tool for robust analysis. The generalized Nyquist plot constructed from the system's eigenloci is a celebrated tool developed in the late 1970s \cite{Desoer:80}. Nevertheless, the eigenloci are known to be neither good robust stability nor robust performance indicators of feedback systems \cite{Maciejowski:89}. One possible remedy to the robustness issue is to adopt the Nquist array with Gershgorin bands \cite[Th.~2.11]{Maciejowski:89}. The principal region \cite{Postlethwaite:81} offers another solution for stable MIMO systems, which is shaped like a frequency-wise curvilinear rectangle. The rectangle, integrated from the principal gain and principal phase notions~\cite{Postlethwaite:81}, provides a bound of the eigenloci. Recently, the operator-theoretic notion of scaled relative graph \cite{Ryu:21} was proposed in \cite{Chaffey:21c, Chaffey:21j} for stable nonlinear systems analysis. The brand-new notion mixes gain and phase values \cite{Chen:21_angle} into a complex scalar, reminiscent of the classical Nyquist plot. The operator-theoretic notion opens the door for revisiting appropriate graphical tools for robust analysis of MIMO systems.

In this paper, inspired by \cite{Chaffey:21j}, we propose a frequency-wise notion of the scaled graph for $p$-dominant MIMO systems that are possibly unstable. The proposed scaled graph contains both gain and phase information of the frequency response matrix of a system, which provides a new frequency-wise bound of the system's eigenloci. It can be used as a robustness measure of $p$-dominance. Our definition is shown to be less conservative than the original definition in \cite{Chaffey:21j}. The main result of this paper concerns dominance analysis of feedback systems: a sufficient condition on separation of the scaled graphs of open-loop systems is provided, guaranteeing dominance preservation of the closed-loop system. We then provide a numerical example for dominance analysis which demonstrates that using the main result can be less conservative than using a result involving the principal region \cite{Postlethwaite:81}.
 
The rest of this paper is organized as follows. In Section~\ref{sec: dominance},  a problem formulation for feedback dominance analysis is proposed. In Section~\ref{sec: SRG}, the scaled graph of matrices is introduced and a matrix separation theorem is presented. We next employ the matrix scaled graph results to MIMO systems in a frequency-wise manner in Section~\ref{sec: systemSRG}. Our main result, a graphical separation condition to guarantee feedback dominance, is presented in Section~\ref{sec: main}. We also specialize the main result to recover several existing results for feedback stability and dominance. Section~\ref{sec:conclusion} concludes this paper. 

\emph{Notation}: Denote by $\cop$ (resp. $\mathbb{C}_-$) the open complex right (resp. left) half-plane. Let $\ccp$ denote the closed right half-plane. Denote by $\mathcal{R}^{\mm}$ the set of $m\times m$ real-rational proper transfer function matrices. Let $\rlinf^{\mm}$ denote the subset of $\mathcal{R}^\mm$ consisting of matrices with no poles on the imaginary axis  $j\mathbb{R}$. The closed subspace of $\rlinf^\mm$ consisting of stable matrices with no poles in $\ccp$ is denoted as $\rhinf^\mm$.  The distance between a point $z\in \mathbb{C}$ and a set $\mathcal{X}\subset \mathbb{C}$ is defined as 
$\dist{z}{\mathcal{X}} = \inf_{x\in \mathcal{X}} \abs{z-x}$.

\section{Dominance of MIMO Systems and 
Problem Formulation}\label{sec: dominance}

\subsection{Definitions}

Dominance of a MIMO system characterizes the level of instability of the system by counting its number of (hyperbolic or strongly) unstable poles as defined below.

\begin{definition}\label{def: open-loop_dominance}
A system $G\in \rlinf^{m\times m}$ of McMillan degree $n$ is said to be \emph{$p$-dominant} if 
\bex 
p~=~\text{the number of poles of}~G ~\text{in}~\cop
\eex 
where $p\in \interval{0}{n}$.
\end{definition}

In the case that $G$ admits a minimal state-space realization $(A, B, C, D)$, an equivalent characterization of Definition~\ref{def: open-loop_dominance} is that the state matrix $A$ has $p$ eigenvalues in $\cop$ and $n-p$ eigenvalues in $\mathbb{C}_-$. Such a state-space characterization was proposed in \cite[Def.~1]{Forni:18} as an original form for dominant systems. One notable specialization of Definition~\ref{def: open-loop_dominance} is asymptotic stability interpreted as $0$-dominance. Specifically, $G\in \rlinf^{m\times m}$ is $0$-dominant if and only if $G\in \rhinf^{m\times m}$. The dominance of a system, as its name implies, indicates the system's asymptotic behavior \cite{Forni:18}. For example, consider a $p$-dominant system $G\in \rlinf^{m\times m}$ whose poles are distinct and decompose the system as $G(s) = \sum_{k=1}^n \frac{H_k}{s-q_k}$, where $H_k$ is the residue matrix with respect to the pole $q_k$. The asymptotic behavior of its impulse response $g(t)$ is determined by $p$ persistent modes, that is, 
$\lim_{t\to \infty} g(t) = \sum_{k=1}^{p}  H_ke^{q_k t}$.

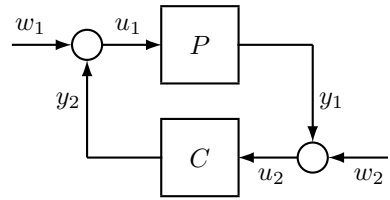
\begin{figure}[htb]
  \centering
  \setlength{\unitlength}{1mm}
  \begin{picture}(50,25)
  \thicklines \put(0,20){\vector(1,0){8}} \put(10,20){\circle{4}}
  \put(12,20){\vector(1,0){8}} \put(20,15){\framebox(10,10){$P$}}
  \put(30,20){\line(1,0){10}} \put(40,20){\vector(0,-1){13}}
  \put(38,5){\vector(-1,0){8}} \put(40,5){\circle{4}}
  \put(50,5){\vector(-1,0){8}} \put(20,0){\framebox(10,10){$C$}}
  \put(20,5){\line(-1,0){10}} \put(10,5){\vector(0,1){13}}
  \put(5,10){\makebox(5,5){$y_2$}} \put(40,10){\makebox(5,5){$y_1$}}
  \put(0,20){\makebox(5,5){$w_1$}} \put(45,0){\makebox(5,5){$w_2$}}
  \put(13,20){\makebox(5,5){$u_1$}} \put(32,0){\makebox(5,5){$u_2$}}
  \end{picture}\caption{A positive feedback system $\goflti$.} \label{fig:feedback}
\end{figure}

The main focus of this paper is dominance analysis of feedback interconnected systems from a frequency-domain perspective. Consider a positive feedback system $\goflti$ illustrated by Fig.~\ref{fig:feedback}, where $P\in \mathcal{R}^\mm$ and $C\in \mathcal{R}^\mm$.  A feedback system $\gof$ is said to be \emph{well-posed} \cite[Lem.~5.1]{Zhou:96} if the matrix $I-P(\infty)C(\infty)$ has full rank. The dominance of the feedback system can be defined below. 
  
\begin{definition}\label{def: feedback_dominance}
  A feedback system $\goflti$ is said to be \emph{$p$-dominant} if it is well-posed and the transfer function matrix  
  \begin{align*}
    \tbt{I}{-C}{-P}{I}^{-1} \coloneqq \tbo{w_1}{w_2}\mapsto  \tbo{u_1}{u_2} 
  \end{align*}
 belongs to $\rlinf^{2m\times 2m}$ and is $p$-dominant.
\end{definition}

\subsection{Problem Formulation}

Over the past half-century, \emph{robust stability} for feedback systems has been extensively studied \cite{Zhou:96}. A simple problem formulation of robust stability analysis can be stated as follows: Given systems $P, C\in \rhinf^{\mm}$, find a condition on $P$ and $C$ such that the feedback system $\gof$ remains stable. The celebrated small gain theorem \cite[Sec.~9.2]{Zhou:96} has been acknowledged to be the most important tool in robust control theory. A frequency-wise version of the small gain condition is given by 
\begin{align}\label{eq: small_gain}
  \overline{\sigma}(P(\jw)) \overline{\sigma}(C(\jw))<1
\end{align}
for all $\omega\in \interval{-\infty}{\infty}$, where $\overline{\sigma}(\cdot)$ denotes the largest singular value of a matrix. Robustness of \eqref{eq: small_gain} can be interpreted in the scenario that $C$ is an \emph{uncertain} system characterized by a gain-bounded uncertainty set:
\begin{align*}\label{eq: gain_uncertainty}
\mathcal{B}_{\delta} = \{\Delta \in \rhinf^{\mm}\mid \overline{\sigma}(\Delta(\jw))\leq \delta(\omega),  \omega\in \interval{-\infty}{\infty} \},
\end{align*}
where $\delta:\interval{-\infty}{\infty} \to \interval[open right]{0}{\infty}$ is a given frequency-wise gain bound. The set $\mathcal{B}_{\delta}$ here describes the ``ball''-type uncertainties. Given $C\in   \mathcal{B}_{\delta}$, condition~\eqref{eq: small_gain} can be adapted to a gain constraint on $P$, that is, $\overline{\sigma}(P(\jw)) < \sbkt{\delta(\omega)}^{-1}$ for all $\omega\in \interval{-\infty}{\infty}$.

We are interested in \emph{robust dominance analysis} of feedback systems. It shares the same flavor with the robust stability analysis mentioned above. Such an analysis is largely inspired from  intrinsic local robustness of rich behaviors in nonlinear biological networks beyond stability. A local model for such a network can naturally be a $p$-dominant system, i.e., $p$ unstable poles in $\cop$. We envision that in this scenario the local dominance or instability should be robust. Conducting local robustness analysis of the network under stable unmodeled dynamics or uncertainties $C$ requires an \emph{explicit answer} to the following question: 

Given $p$-dominant $P\in \rlinf^{\mm}$ and $C\in \rhinf^{\mm}$, when does the feedback system $\gof$ remain $p$-dominant? 

More generally, one may remove the stability hypothesis on $C$ and further quantify $C$ by its dominance. To this end, we can formulate the following problem which is at the core of robust dominance analysis in this paper. 

\begin{problem}\label{problem}
Given $p_1$-dominant $P\in \rlinf^{\mm}$ and $p_2$-dominant  $C\in \rlinf^{\mm}$, find a criterion on $P$ and $C$ such that the feedback system $\gof$ is $(p_1+p_2)$-dominant.
\end{problem}

Our solution (Theorem~\ref{thm:lti_instability}) to Problem~\ref{problem} will be developed in a \emph{graphical} and \emph{frequency-wise} sense. With a frequency-domain approach, a crucial step toward the solution is to determine whether $\det(I-AB)\neq 0$ for given $A, B\in \cmm$. We will examine the invertibility problem of $I-AB$ via the recent notion of matrix scaled graph elaborated below. 

\section{The Scaled Graph of Matrices}\label{sec: SRG}
For $x, y\in \cm$, denote the real inner product by $\ininf{x}{y}\coloneqq \mathrm{Re}(x^* y)$ and the Euclidean norm by $\norm{x}_2\coloneqq\sqrt{x^*x}$. Given $x, y\in \cm$,  define the  \textit{gain} $\gamma(x, y)\in \interval{0}{\infty}$ from $x$ to $y$ by
\bex
\gamma(x, y)\coloneqq \frac{\norm{y}_2}{\norm{x}_2}\quad \text{if}~x \neq 0
\eex 
and $\gamma(x, y)= \infty$, otherwise. Moreover, define the \textit{phase} $\theta(x, y)\in \interval{0}{\pi}$ from $x$ to $y$ by
\bex 
\theta(x, y)\coloneqq \arccos \dfrac{\ininf{x}{y}}{\norm{x}_2\norm{y}_2}\quad \text{if}~x, y\neq 0,
\eex
and $\theta(x,y)=\emptyset$, otherwise. Notice that $\gamma(y, x)=1/\gamma(x, y)$ and $\theta(x, y)=\theta(y, x)$ for all nonzero $x, y\in \cm$.

We next introduce a recent graphical tool called the scaled graph \cite{Chaffey:21j, Ryu:21, Pates:21} for visualizing a matrix in the complex plane. For a matrix $A\in \cmm$, the \emph{scaled graph} $\sg(A)\subset \bar{\mathbb{C}}$ is defined to be 
\begin{multline}\label{eq: matrix_sg} 
\scriptsize
  \sg(A)\coloneqq \left\{ z=\gamma(x, Ax)  e^{\pm j\theta(x, Ax)} \in \bar{\mathbb{C}} \mid \right. \\
  \left. \vphantom{e^{\pm j\theta(x, Ax)}} 0\neq x\in \cm,  Ax \neq 0 \right\}.
\end{multline}
Note that $\sg(A)$ contains the gain $\gamma(\cdot, \cdot)$ and phase $\theta(\cdot, \cdot)$ information of the input-output pairs of $A$ and it is symmetric with respect to the real axis due to the term ${\pm j \theta(\cdot, \cdot)}$. It has been shown in \cite[Th.~1]{Pates:21} that the scaled graph has a close connection with a well-known graphical notion for matrices, the numerical range. This gives an algorithmic method for plotting a scaled graph, by using existing algorithms to plot a related numerical range.  

By swapping the role of input-output pairs in contrast to $\sg(A)$, we define the \textit{inverse scaled graph} $\sg^{\dagger}(A)$ as 
 \begin{multline}\label{eq: matrix_isg} 
\sg^{\dagger}(A)\coloneqq \left\{z=\gamma(Ax, x) e^{\pm j \theta(Ax, x)} \in \bar{\mathbb{C}} \mid  \right. \\
\left. \vphantom{  e^{\pm j \theta(Ax, x)}} 0\neq x\in \cm,  Ax\neq 0 \right\}.
 \end{multline}

The matrix scaled graph is invariant to unitary similarity transformations, that is,
\begin{equation}\label{eq: invariance}
  \sg(U^*A U) = \sg(A)
\end{equation}
 for any unitary matrix $U\in \cmm$. A simple proof reads as follows:
 \begin{align*}
  \gamma(x, U^* A Ux)&=\gamma(Ux, A Ux)  =\gamma(y, Ay), \\ 
  \theta(x, U^* A Ux)&= \theta(Ux, AUx)  =\theta(y, Ay),
 \end{align*}
 where $y\coloneqq Ux \in \mathbb{C}^m$ can be arbitrary. Consider the Schur decomposition $A=U^* T U$, where $T\in \cmm$ is upper triangular. In light of the invariance property in \eqref{eq: invariance}, note that plotting $\sg(A)$ boils down to plotting $\sg(T)$. 
 
Our definition \eqref{eq: matrix_sg} of scaled graphs is different from the original definition in \cite{Ryu:21} in the sense that we adopt the complex space $\cm$ rather than $\mathbb{R}^m$. Our definition coincides with the definition proposed in \cite{Pates:21}. Adopting the complex space is more natural for the purpose of frequency-domain analysis of LTI systems in Section~\ref{sec: systemSRG}. Another  benefit that comes with definition \eqref{eq: matrix_sg} is that the scaled graph of a matrix provides a bound for all its nonzero eigenvalues, regardless of the dimension $m$, as elaborated in Lemma~\ref{lem: sg_eigen}. Such a property does not hold when $m=2$ under the original definition; see \cite[Th.~1]{Ryu:21}.
\begin{lemma}\label{lem: sg_eigen}
  For $A\in \cmm$, it holds that \bex \lambda_i(A)\in \sg(A)\eex for all $i=1, 2, \ldots, m$ and $\lambda_i(A)\neq 0$.  
\end{lemma}
\begin{proof}
 Let $\lambda\in \mathbb{C}$ be any nonzero eigenvalue of $A$ associated with an eigenvector $x\in \cm$. Based on the input vector $x$, we generate the corresponding point in $\sg(A)$:
\begin{align*}
  z&= \frac{\norm{\lambda x}_2}{\norm{x}_2} \exp\rbkt{ j\arccos\frac{\rep(x^* \lambda x)}{\norm{x}_2\norm{\lambda x}_2}}\\
  &= \abs{\lambda}\exp\rbkt{j \arccos\frac{\rep(\lambda)}{\abs{\lambda}}} = \abs{\lambda} \exp\rbkt{ j \abs{\angle{\lambda}}}.
\end{align*}
Therefore, we conclude that $\lambda \in  \{z, \bar{z}\}\subset \sg(A)$.
\end{proof}

A generalized version of Lemma~\ref{lem: sg_eigen} for linear operators has been recently established in {\cite[Th.~1(ii)]{Pates:21}}, while we provide a short proof above for completeness.

Given two matrices $A, B$, we show that the scaled graphs of $A$ and $B$ can be adopted to determine the full-rankness of $I-AB$.  The following \emph{matrix separation theorem} lays the foundation for dominance analysis of MIMO systems later. 
\begin{theorem}\label{lem:matrix}
  For $A, B\in \cmm$, if 
  \be\label{eq:lem_separation_sg}
  \sg(A) \cap\sg^\dagger(B)= \emptyset,
  \ee 
  then $\mathrm{det}(I-AB)\neq 0$.
\end{theorem}
\begin{proof}
    See Appendix. 
\end{proof}

One important observation from Theorem~\ref{lem:matrix} is that under the separation condition \eqref{eq:lem_separation_sg}, the $m$ eigenvalues of $AB$ never intersect with the critical point ``$1$''. This property is reminiscent of the classical Nyquist stability approach when a \emph{positive} feedback system is under consideration. Thanks to \eqref{eq: invariance}, condition~\eqref{eq:lem_separation_sg} is also robust under any unitary similarity transformation perturbation, namely, 
\bex 
\sg(U^* A U) \cap  \sg^\dagger(V^*BV) =  \sg(A) \cap\sg^\dagger(B)
\eex 
for all unitary matrices $U$ and $V$. 
Inspired by the analogy above, we propose to adopt the scaled graph as an alternative to the generalized Nyquist plot \cite{Desoer:80} of MIMO systems. 

\section{The Scaled Graph of MIMO Systems}\label{sec: systemSRG} 

Consider a MIMO LTI system $G\in \rlinf^{m\times m}$ of McMillan degree $n$. Such an system does not have any pole on the extended imaginary axis $j\mathbb{R}\cup\{\infty\}$. For all $\omega\in \interval{-\infty}{\infty}$, we define the (frequency-wise) scaled graph of $G$ to be 
\begin{multline}\label{eq: system_sg}
 \hspace{-3mm} \sg(G(\jw)) \coloneqq \left\{ z(\jw)=\gamma(x, G(\jw)x)  e^{\pm j\theta(x, G(\jw)x)} \in \bar{\mathbb{C}}  \right. \\
\left.  \mid \vphantom{e^{\pm j\theta(x, G(\jw)x)}}  0\neq x\in \cm,  G(\jw)x \neq 0 \right\}.
\end{multline}
The  (frequency-wise) inverse scaled graph $\sg^\dagger(G(\jw))$ can be defined similarly to \eqref{eq: matrix_isg}. In the case when $G$ is SISO, the scaled graph of $G$ recovers the Nyquist plot of $G$. 
\begin{proposition}\label{prop: sg_nyquist}
For $G\in \rlinf^{1\times 1}$, the two sets are equal:
\begin{enumerate}
 \renewcommand{\theenumi}{\textup{(\roman{enumi})}}\renewcommand{\labelenumi}{\theenumi}
\item The Nyquist plot: $\{G(\jw)\mid \omega\in \interval{-\infty}{\infty} \}$;
\item The union of the scaled graphs over $\omega\in \interval{-\infty}{\infty}$: $\{\sg(G(\jw))\mid \omega\in \interval{-\infty}{\infty}\}$.
\end{enumerate}
\end{proposition}

\begin{proof}
For all nonzero $x\in \mathbb{C}$, we have
\begin{align*}
\gamma(x, G(\jw)x)&=\frac{\abs{G(\jw)}\abs{x}}{\abs{x}}=\abs{G(\jw)},  \\
\theta(x, G(\jw)x)&=\arccos\frac{\rep( G(\jw))\abs{x}^2}{\abs{x}\abs{G(\jw)}\abs{x}}=\abs{\angle G(\jw)},
\end{align*}
where $\omega\in \interval{-\infty}{\infty}$. By definition~\eqref{eq: system_sg}, at each $\omega$, the set $\sg(G(\jw)) = \{{G(\jw)}, \overline{G(\jw)} \}$ contains a conjugate pair. It follows from the conjugate symmetry $G(-\jw)=\overline{G(\jw)}$ that $\{ \sg(G(\jw))\mid \omega\in \interval{-\infty}{\infty} \} = \{G(\jw)\mid \omega\in \interval{-\infty}{\infty} \}$.
\end{proof}

For a MIMO system $G\in \rlinf^\mm$, the \emph{generalized Nyquist plot} of $G$ is depicted by the \emph{$m$ continuous eigenloci}:
\bex 
\left\{ \lambda_i(G(\jw))  \mid \omega \in \interval{-\infty}{\infty}, i=1, 2, \ldots, m \right\},
\eex 
based on which the generalized Nyquist criterion \cite{Desoer:80} was successfully formulated half a century ago. However, it is well known that taking the eigenloci directly for systems analysis will unavoidably result in various drawbacks. Firstly, a practical system may not be precisely known due to its unmodeled dynamics.  Requiring precise information of the eigenloci of the system turns out to be \emph{highly unrealistic}. Secondly, it has been widely acknowledged that the eigenloci do not give reliable information -- they are neither good robust stability indicators nor robust performance indicators of feedback systems. A small perturbation of a system can lead to a large change in its eigenloci. Robust control is targeted at this scenario, where the system is often assumed to belong to a set of uncertain systems under some specific perturbations. The underlying requirement behind an appropriate uncertainty description is that the description has to \emph{bound the eigenloci}. For instance, the small gain analysis endowed with the set of gain-bounded uncertain systems is known to be at the heart of robust control theory. The gain-based analysis in \eqref{eq: small_gain} is owing to the eigenlocus-gain bound:
\bex 
\abs{\lambda_i(G(\jw))}\leq \overline{\sigma}(G(\jw)),
\eex 
where $i=1, 2, \ldots, m$. Nevertheless, using only the gain information for systems analysis is known to be highly conservative. Is it possible to subsume both the gain and phase information into systems analysis? The answer is affirmative. Owing to Lemma~\ref{lem: sg_eigen}, the scaled graph offers a new bound for the generalized Nyquist plot from a mixed gain/phase perspective. This indicates that the scaled graph may be viewed as a \emph{robust alternative} to the generalized Nquist plot for MIMO systems.
\begin{proposition} 
  For $G\in \rlinf^{m\times m}$, it holds that
  \bex 
 {\lambda_i(G(\jw))} \in \sg(G(\jw)),
\eex 
where  $i=1, 2, \ldots, m$ and $\lambda_i(G(\jw))\neq 0$.
\end{proposition}
\begin{proof}
  The result follows directly from Lemma~\ref{lem: sg_eigen} in a frequency-wise manner. 
\end{proof}

\begin{example}
  Consider a $3\times 3$ system given by
    \bex\label{eq:example}
    G(s)=  U^*(s)\mathrm{diag}\rbkt{\frac{3s+6}{s^2+s+1},\frac{2s+2}{s-1},\frac{s+10}{s^2-2s-2}} U(s),
    \eex
  where $U(s)$ can be any stable all-pass system satisfying that $U^*(\jw) U(\jw)=I$ for all $\omega \in \interval{-\infty}{\infty}$. Clearly $G$ is a $2$-dominant system. The scaled graph and generalized Nyquist plot of $G$ for $\omega\in \interval{0}{100}~\mathrm{rad/s}$ are sketched in Fig.~\ref{fig: exmp}.  
\end{example}

\begin{figure}[htb]
  \centering
   \includegraphics[width=0.8\linewidth, trim={0.6cm 0.6cm 0.6cm 0.6cm}]{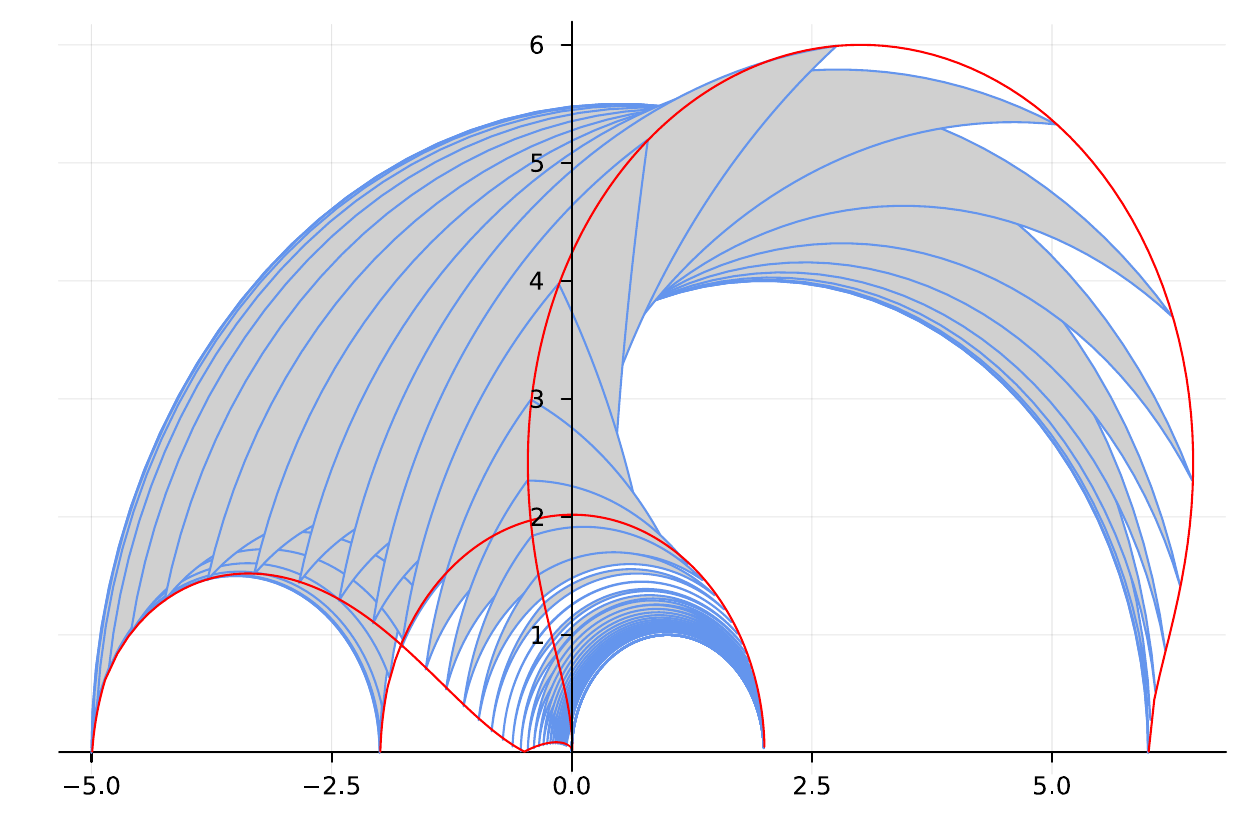}
 \caption{A sketch of the scaled graph of $G$ over $\omega\in \interval{0}{100}~\mathrm{rad/s}$ in the upper half complex plane, within which the generalized Nyquist plot ($3$ continuous eigenloci) of $G$ lies (red).}\label{fig: exmp} 
\end{figure}

\begin{remark}
 For a system $G\in \rlinf^{m\times m}$, the scaled graph $\sg(G(\jw))$ is \emph{less conservative} than the recent operator scaled graph ${\mathbf{SG}}(G)$ proposed in \cite{Chaffey:21j, Pates:21} due to the following two reasons. First, the former is a frequency-wise definition, while the latter is not. For a SISO system $G$, it has been proved in \cite[Th.~4]{Chaffey:21j} that ${\mathbf{SG}}(G)$ is the hyperbolic convex hull of the Nyquist plot of $G$. In light of Proposition~\ref{prop: sg_nyquist}, one can immediately obtain that
 \bex 
 \{ \sg(G(\jw))\mid \omega\in \interval{-\infty}{\infty} \}\subset {\mathbf{SG}}(G),
 \eex 
 where the set inclusion relation is \emph{strict}. Second, $\sg(G(\jw))$ is perfectly suited for studying possibly unstable LTI systems, e.g. $G(s)=\frac{1}{1-s}$. This shares the same flavor with the Nyquist plot since both of them make full use of the frequency-wise nature of LTI systems with \emph{no regard to stability or instability}. By contrast, ${\mathbf{SG}}(G)$ is an operator-theoretic concept. For unstable systems, how to define, interpret and sketch ${\mathbf{SG}}(G)$ still remains unclear.
\end{remark}

\begin{remark}
 When defining $\sg(G(\jw))$ in \eqref{eq: system_sg}, the premise of restricting $G\in \rlinf^\mm$ is due to a technical reason for the sake of brevity of our presentation. Note that $G(\jw)$ is a well-defined matrix since $G$ has no pole on $j\mathbb{R}\cup \{\infty\}$. This premise can be \emph{removed}. We can generally allow $G\in \mathcal{R}^\mm$ in \eqref{eq: system_sg}. In such a case, we have to take semicircular indentations around all the $\jw$-axis poles of $G$ and construct the indented Nyquist contour $\Gamma$. Based on the contour, we can define $\sg(G(s))$ for all $s\in \Gamma$ analogously to definition~\eqref{eq: system_sg} for all $\omega$. This is known to be a standard remedy in frequency-domain analysis of LTI systems; see \cite{Desoer:80,Chen:23arxiv}. Finally, we remark that our main result (Theorem~\ref{thm:lti_instability}) proposed below can be tailored to fit such a generalized definition.
\end{remark}

\section{Main Results}\label{sec: main}
This section presents our main result for feedback dominance analysis. Under a separation condition characterized by scaled graphs, dominance of a closed-loop system can be deduced from that of open-loop components. It is shown that several existing criteria for feedback stability and dominance analysis, e.g., the small gain theorem and passivity theorem, can be recovered from the main result. A comparison is made between the scaled graph and the principal region \cite{Postlethwaite:81}.

\subsection{A Graphical Dominance Theorem}
To resolve Problem~\ref{problem}, we aim at deriving a sufficient condition on $P$ and $C$ by using their scaled graphs to determine the dominance of the feedback system $\goflti$. 

An assumption on feedback systems is required before we move on. A feedback system $\gof$ is said to have \emph{no unstable pole-zero cancellation} if the number of poles of $PC$ in $\ccp$ is equal to the sum of that of $P$ in $\ccp$ and that of $C$ in $\ccp$. Throughout, all feedback systems are reasonably assumed to be free of unstable pole-zero cancellations since otherwise a feedback loop can generate unstable internal signals, which is highly undesirable in practice. Under this assumption, we can derive the following simple characterization of Definition~\ref{def: feedback_dominance}.
  
\begin{lemma}\label{lem: dominance_feedback}
  Consider $P, C\in\rlinf^{m\times m}$ and suppose that there is no unstable pole-zero cancellation in $\goflti$. Then the following two statements are equivalent:
  \begin{enumerate} \renewcommand{\theenumi}{\textup{(\roman{enumi})}}\renewcommand{\labelenumi}{\theenumi}
    \item  The feedback system $\goflti$ is $p$-dominant.
    \item  The transfer matrix $(I-PC)^{-1}$  is $p$-dominant.
   \end{enumerate}
\end{lemma}
\begin{proof}
The proof follows from the similar reasons as those in \cite[Th.~5.7]{Zhou:96} and is omitted for brevity.
\end{proof}

We now state the main result of this paper below as an answer to Problem~\ref{problem}.
\begin{theorem}[SG dominance theorem]\label{thm:lti_instability}
Let $P\in\rlinf^{m\times m}$ be $p_1$-dominant and $C\in\rlinf^{m\times m}$ be $p_2$-dominant. Suppose that there is no unstable pole-zero cancellation in $\goflti$. Then $\goflti$ is $(p_1+p_2)$-dominant if 
\be\label{eq:large_sg_frequency}
    \sg(\tau P(\jw))\cap\sg^\dagger(C(\jw))=\emptyset  
\ee 
for all $\omega\in \interval{-\infty}{\infty}$ and  $\tau \in \interval[open left]{0}{1}$.
\end{theorem}
\begin{proof}
    See Appendix.
\end{proof}

Theorem~\ref{thm:lti_instability} provides a graphical test for feedback dominance analysis from a frequency-sweep perspective. The test boil downs to computation of the scaled graphs of complex matrices. Such computation can be performed from either numerical sampling or the established connection \cite[Th.~1]{Pates:21} with the computation of numerical ranges. For SISO systems, note that condition~\eqref{eq:large_sg_frequency} reduces to 
\bex 
\tau P(\jw)C(\jw) \eqqcolon \tau L(\jw) \neq 1
\eex 
for all $\omega\in \interval{-\infty}{\infty}$ and $\tau\in \interval[open left]{0}{1}$. The smallest distance between $\tau L(\jw)$ and ``$1$''  provides a robustness indicator for $p$-dominance. Thanks to \eqref{eq:large_sg_frequency}, we now have a MIMO version of such a robustness measure:
\begin{definition} 
The $p$-dominance margin of a feedback system $\gof$ is defined by
  \bex 
  \mathrm{dm}(\gof)\coloneqq \inf_{\substack{z\in \sg(\tau P(\jw))\\ \omega\in \interval{-\infty}{\infty}, \tau \in \interval[open left]{0}{1}}} \mathrm{dist}(z, \sg^\dagger(C(\jw))).
  \eex 
\end{definition}
For SISO nonlinear Lur'e feedback systems, an alternative definition of $p$-dominance margins was studied in \cite{Padoan:19} using the circle criterion. 

We next elaborate on the differences between Theorem~\ref{thm:lti_instability} and an existing operator scaled graph result  \cite[Th.~2]{Chaffey:21j}. 
\begin{enumerate}
\renewcommand{\theenumi}{\textup{(\roman{enumi})}}\renewcommand{\labelenumi}{\theenumi}
  \item Theorem~\ref{thm:lti_instability} aims at feedback dominance analysis, while \cite[Th.~2]{Chaffey:21j} is a feedback stability result. How to generalize \cite[Th.~2]{Chaffey:21j} for dominance analysis is nontrivial.
  \item When restricting open-loop systems in \cite[Th.~2]{Chaffey:21j} to be LTI, Theorem~\ref{thm:lti_instability} is less conservative in feedback stability analysis, as partially explained in Example~\ref{exmp: SG_literature}.
\end{enumerate}

The following example indicates that Theorem~\ref{thm:lti_instability} can be generally less conservative than operator-theoretic separation results for feedback dominance.

\begin{example}\label{exmp: SG_literature}
Consider $2$-dominant $P(s)=\frac{1}{(s-1)^2(s+2)}$ and $1$-dominant $C(s)=\frac{2s+2}{s-2}$. The feedback system $(I-P(s)C(s))^{-1} P(s)= \frac{s-2}{s^4-2s^3-3s^2+6s-6}$ is $3$-dominant. One can observe that the Nyquist plot of $P$ and that of ${C}^{-1}$ have an ``intersection'' at $\frac{1}{2}$ from an operator-theoretic viewpoint. However, according to Theorem~\ref{thm:lti_instability}, the feedback system is indeed $3$-dominant and the ``intersection'' here is quite misleading. In fact, from a frequency-wise perspective, the ``intersection'' disappears due to $\sg^\dagger(C(\infty))= \sg(P(0))=\frac{1}{2}$ at \emph{different} frequencies $0$ and $\infty$. Such a case is hard to be exposed in operator-theoretic settings.
\end{example}

\subsection{Robust Dominance Analysis under Stable Uncertainties}

Insipred by the small gain analysis \cite[Ch.~9]{Zhou:96}, we present a specialization below to demonstrate how to leverage Theorem~\ref{thm:lti_instability} for robust dominance analysis. Let $\mathbb{S}_k\subset \mathbb{C}$ be a simply connected closed region such that $\mathbb{S}_k$ is symmetric with respect to the real axis. The set of all such regions in $\mathbb{C}$ is denoted by $\mathcal{S}=\cup_{k} \mathbb{S}_k$. Define a region-valued mapping by $\rho: \interval{-\infty}{\infty} \to \mathcal{S}$. Given a mapping $\rho$, we define the set of stable and minimum-phase uncertainties: 
\begin{align*}
  \mathbf{\Delta}_{\rho}^{m\times m}\coloneqq \{ \Delta \in \rhinf^{m\times m} \mid \sg(\Delta(\jw))\subset \rho(\omega) \}.
\end{align*}
Intuitively, the scaled graph of $\Delta$ is supposed to be covered by a known frequency-wise mapping $\rho$. Next, consider an uncertain system $P_\Delta$ subject to \emph{additive} uncertainties:
\begin{align*}
  P_\Delta(s)= P(s) + \Delta(s), 
\end{align*}
where $P \in \mathcal{R}^{m\times m}$ is known and $\Delta \in \mathbf{\Delta}_{\rho}^{m\times m}$. Consider an uncertain feedback system $P_\Delta\,\#\, C$, where $C\in \mathcal{R}^{m\times m}$ is known. Such a loop $P_\Delta\,\#\, C$ can be reorganized into a standard one $\Delta\,\#\, G$, where $G\coloneqq (I-CP)^{-1}C$ is the nominal feedback system which is known. Hence, the inverse scaled graph of $(I-CP)^{-1}C$ becomes available and can be handily frequency-wise sketched. Applying Theorem~\ref{thm:lti_instability} to $\Delta\,\#\, G$ yields the following corollary for robust analysis:
\begin{corollary}\label{cor: additive}
Suppose that $\Delta \in \mathbf{\Delta}_{\rho}^{m\times m}$ and the nominal feedback system $G=(I-CP)^{-1}C$ is $p$-dominant. The uncertain feedback system $P_\Delta\,\#\, C$ is $p$-dominant if 
$$\inf_{z\in \sg^\dagger(G(\jw))} \mathrm{dist}(z, \tau \rho(\omega))>0$$
for all $\omega \in \interval{-\infty}{\infty}$ and $\tau \in \interval{0}{1}$.
\end{corollary}
\begin{proof}
    The proof follows directly from Theorem~\ref{thm:lti_instability}.
\end{proof}

Similar statements to Corollary~\ref{cor: additive} apply to other types of uncertainties, such as multiplicative uncertainties.

\subsection{Specialization of the Main Result}

Feedback stability can be viewed as a $0$-dominance problem. As a consequence, a useful graphical test on feedback stability can be derived handily from Theorem~\ref{thm:lti_instability}.
\begin{corollary}\label{thm:lti_stability}
For $P, C \in \rhinf^{\mm}$, the feedback system $\gof$ is stable if 
  \bex 
  \sg(\tau P(\jw))\cap\sg^\dagger( C(\jw))=\emptyset  
    \eex 
    for all $\omega\in \interval{-\infty}{\infty}$ and  $\tau \in \interval[open left]{0}{1}$.
\end{corollary}
 \begin{proof}
The proof is completed by invoking Theorem~\ref{thm:lti_instability} with $p_1=p_2=0$.
 \end{proof}

The scaled graph with Theorem~\ref{thm:lti_instability} and Corollary~\ref{thm:lti_stability} can be further interpreted from gain, phase and passivity perspectives. Specifically, given $P\in \rhinf^{\mm}$, the maximum singular value of $P(\jw)$ can be \emph{read out} from $\sg(P(\jw))$ via 
$$
\overline{\sigma}(P(\jw)) =  \sup_{z\in \sg(P(\jw))} \abs{z},
$$
i.e., the maximum radius at each frequency. This is in contrast to the operator scaled graph ${\mathbf{SG}}(P)$ in 
\cite{Chaffey:21j}, whose maximum radius is the induced $2$-norm (aka $\mathcal{H}_\infty$-norm) of $P$. The maximum phase of $P(\jw)$ (or the singular angle \cite{Chen:21_angle}) can be similarly read out via 
$$
\psi(P(\jw)) = \sup_{z\in \sg(P(\jw))} \abs{\angle{z}},
$$
i.e., the maximum absolute value of the argument at each frequency. We may also translate positive realness (aka passivity) \cite[Sec.~2.7]{Anderson:73} into a scaled graph description. A system $P$ is called positive real (resp. strongly positive real) if $P(\jw)+P(\jw)^*\geq 0$ (resp. $>0$) for all $\omega \in \interval{-\infty}{\infty}$. In this case, a right half-plane description $\sg(P(\jw))\subset \ccp$ (resp. $\sg(P(\jw))\subset \cop$) can be obtained. Note that the above graphical descriptions of the gain, phase and passivity can be naturally extended to unstable systems in $\rlinf^\mm$. Historically they were called the conditional gain ($\mathcal{L}_\infty$-norm) \cite{Takeda:73} and generalized positive realness \cite{Anderson:68} as explained in Example~\ref{ex: intro}. This leads to the following result.

\begin{corollary}\label{thm: gain_phase_passivity}
Let $P\in\rlinf^{m\times m}$ be $p_1$-dominant and $C\in\rlinf^{m\times m}$ be $p_2$-dominant. Suppose that there is no unstable pole-zero cancellation in $\goflti$. Then $\goflti$ is $(p_1+p_2)$-dominant if for all $\omega\in \interval{-\infty}{\infty}$, any of the following conditions holds:
  \begin{enumerate}
    \renewcommand{\theenumi}{\textup{(\roman{enumi})}}\renewcommand{\labelenumi}{\theenumi}
   \item \label{item: small_gain} a small gain condition:
   $\overline{\sigma}(P(\jw)) \overline{\sigma}(C(\jw))<1;$ 
\item \label{item: small_phase} a small phase condition: $\psi(P(\jw)) + \psi(-C(\jw))<\pi;$
\item a positive real condition: 
$P(\jw) + P(\jw)^* \geq  0$ and $C(\jw) + C(\jw)^* <0$.
\end{enumerate}
\end{corollary}
\begin{proof}
The proof is omitted since it can be established from the graphical interpretations of the three conditions.
\end{proof} 

A related condition of $0$-dominance in Corollary~\ref{thm: gain_phase_passivity} has been developed independently in \cite{Baron-Prada:25}. Unstable systems are explicitly excluded from the analysis in \cite{Baron-Prada:25}. By contrast, our result represents a significant generalization, allowing unstable open-loop and closed-loop systems.

\subsection{Links to the Principal Region}
We next elaborate on the differences between the scaled graph and the {principal region} \cite{Postlethwaite:81}, a pioneering concept introduced by the British School half a century ago. The principal region is a frequency-wise curvilinear rectangle drawn from the principal gain and principal phase.

We start with some preliminaries on the matrix principal gain and principal phase. Consider the polar decomposition of a matrix $A\in \cmm$:
 \begin{align}\label{eq: polar}
  A = U Q,
 \end{align}
where $U\in \cmm$ is unitary and $Q\in \cmm$ is positive semi-definite uniquely determined by $Q= \sqrt{A^*A}$. It is known that \eqref{eq: polar} is related to the singular value decomposition $A  =  V \Sigma W^*$, where $U = V W^*$ and $Q = W \Sigma W^*$. The $m$ \emph{principal gains} of $A$ are defined to be the $m$ eigenvalues of $Q$: $\sigma_i(A)= \lambda_i(Q) = \lambda_i(\sqrt{A^* A})$,
where $i=1, 2, \ldots, m$, that is, the singular values of $A$. The $m$ \emph{principal phases} of $A$ are defined as the arguments of the $m$ eigenvalues of $U$: $\phi_i(A) = \angle \lambda_i(U) = \angle \lambda_i(V W^*)$,
where $i=1, 2, \ldots, m$. Denote by $\overline{\sigma}(A)$ and $\underline{\sigma}(A)$ the 
\emph{largest and smallest principal gain} of $A$, respectively, and by $\overline{\phi}(A)$ and $\underline{\phi}(A)$ the \emph{largest and smallest principal phase} of $A$.

For the analysis purpose of a MIMO system, utilizing the quadruplet $\overline{\sigma}(\cdot)$, $\underline{\sigma}(\cdot)$, $\overline{\phi}(\cdot)$ and $\underline{\phi}(\cdot)$ is sufficient to bound the system's eigenloci. Concretely, for $\omega \in \interval{0}{\infty}$, the \emph{principal region} of $G\in \rhinf^{\mm}$ is defined to be the following set: 
\begin{multline}\label{eq: principal gain_phase_1}
\hspace{-3mm}  \mathrm{PR}(G(\jw)) \coloneqq 
  \left\{ z\in \mathbb{C}\mid \abs{z}\in \interval{\underline{\sigma}(G(\jw))}{\overline{\sigma}(G(\jw))},  \right.\\
    \left. \angle z \in \interval{\underline{\phi}(G(\jw))}{\overline{\phi}(G(\jw))}  \right\}, 
\end{multline}
if $\overline{\phi}(G(\jw))-\underline{\phi}(G(\jw))<\pi$; 
\begin{align}\label{eq: principal gain_phase}
 \hspace{-2.5mm} \mathrm{PR}(G(\jw)) \coloneqq 
  \left\{ z\in \mathbb{C}\mid \abs{z}\in \interval{\underline{\sigma}(G(\jw))}{\overline{\sigma}(G(\jw))} \right\},  
\end{align}
otherwise. At each frequency $\omega$, intuitively the principal region $\mathrm{PR}(G(\jw))$ shapes like a \emph{curvilinear rectangle} whenever the phase spread constraint $\overline{\phi}(G(\jw))-\underline{\phi}(G(\jw))<\pi$ is satisfied and otherwise it shrinks to an \emph{annular region} determined by the principal gains only; see \cite[Fig.~4]{Postlethwaite:81}. It was shown in \cite[Ths.~1-2]{Postlethwaite:81} that the principal region $\mathrm{PR}(G(\jw))$ contains the $m$ eigenloci of $G$, that is,
\begin{align*}
  \lambda_i(G(\jw)) \in \mathrm{PR}(G(\jw))
\end{align*}
for all $\omega \in \interval{-\infty}{\infty}$ and $i=1, 2, \ldots, m$. Hence, the principal region may also be viewed as a robust version of the generalized Nyquist plot \cite[Sec.~V]{Postlethwaite:81}. 

We next pinpoint the major differences between the scaled graph and principal region. Firstly, the scaled graph is defined for systems in $\rlinf^\mm$, while the principal region is often restricted to $\rhinf^\mm$. Secondly, the phase information in the scaled graph and that in the principal region are different due to distinct phase definitions. One of the restrictions of the principal region is that for some frequencies, it can reduce to an annular region as in \eqref{eq: principal gain_phase}. Lastly, the gain/phase integration in the scaled graph comes from a complex scalar generated by a pair of inputs and outputs, while the principal region adopts a curvilinear rectangle which can bring extra conservatism in systems analysis. Here is an example showing that the scaled graph can be less conservative in feedback dominance analysis.

\begin{example}\label{ex:1d_example}
Consider a feedback system $G\,\#\, (-I)$, where 
$$
G(s)=\tbt{\frac{1}{(s+1)^2}}{\frac{1}{s+1}}{0}{\frac{0.9}{(s+1)^2(s-1)}}$$
is $1$-dominant, with scaled graph shown in Figure~\ref{fig:1d_example}. By Theorem~\ref{thm:lti_instability}, the feedback system is also $1$-dominant, as the scaled graph of $G$ does not intersect the point ``$-1$''.

\begin{figure}
    \centering
\includegraphics[width=0.8\linewidth, trim={0.6cm 0.6cm 0.6cm 0.6cm}]{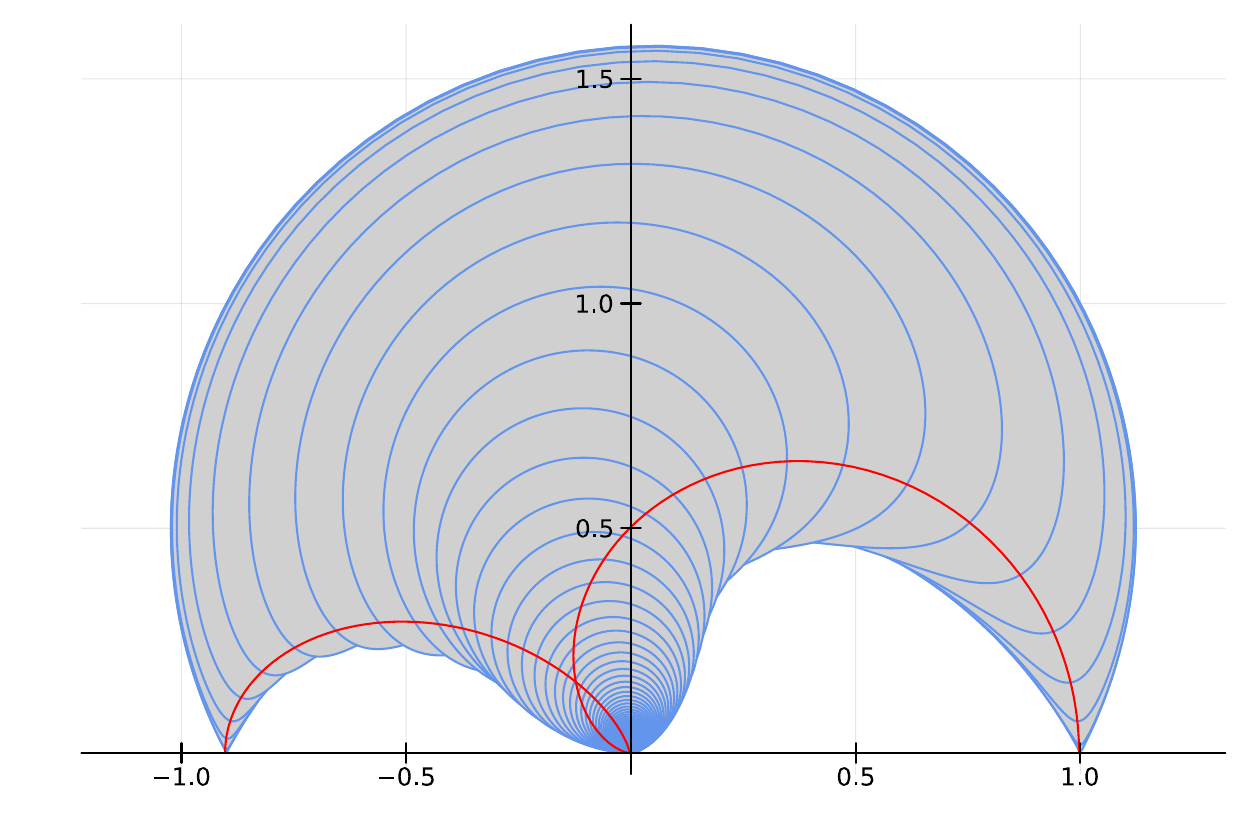}
    \caption{The scaled graph of $G$ in Example~\ref{ex:1d_example} in the upper half-plane, with the generalized Nyquist plot in red.}
    \label{fig:1d_example}
\end{figure}
\end{example}

For a fair comparison, notice that one may first extend the principal region definition in~\eqref{eq: principal gain_phase_1} and \eqref{eq: principal gain_phase} to $p$-dominant systems. Even so, if we examine the principal region, rather than the scaled graph, we cannot draw the same conclusion for $1$-dominance.  At DC, we have
\begin{equation*}
    G(0) = \tbt{1}{1}{0}{-0.9}.
\end{equation*}
The principal phases of $G(0)$ are $\overline{\phi}(G(0)) = \pi$ and $\underline{\phi}(G(0)) = 0$, and the principal gains are $\overline{\sigma}(G(0)) = 1.57$ and $\underline{\sigma}(G(0)) = 0.57$.  The resulting annulus contains the point ``$-1$'', so we cannot conclude anything about the encirclements of the point ``$-1$'' by the eigenloci of $G$ from the plot of the principal region.

\section{Conclusion}\label{sec:conclusion}
In this paper, we defined a frequency-wise notion of the scaled graph for $p$-dominant MIMO LTI systems. The proposed scaled graph sketched in a complex plane provides an alternative and robust version of the generalized Nyquist plot. The proposed main result is a graphical dominance criterion for feedback systems analysis which relies on separation of the scaled graphs of open-loop $p$-dominant systems. A detailed comparison was further made between the scaled graph and the existing principal region.

Our proposed framework provides a fresh perspective on systems robust dominance analysis. It is our hope that the framework opens a door for conducting local robustness analysis of complicated behaviors in nonlinear biological networks subject to dynamic uncertainties. How to adapt the framework to robust analysis of neuromorphic and spiking control systems \cite{Sepulchre_22spiking} is currently under investigation.

\section*{Acknowledgment}
This work was initiated and partially completed when the first author was affiliated with the Department of Electrical Engineering (ESAT), KU Leuven, Belgium, and the second author was affiliated with the Department of Engineering, University of Cambridge, UK and supported by Pembroke College, University of Cambridge, UK. The first author would like to express special thanks to Sei Zhen Khong for useful discussions.

\appendix
 
\emph{Proof of Theorem~\ref{lem:matrix}}:~ 
It suffices to show that the block matrix $\stbt{I}{B}{A}{I}$ has full rank. Let $u=\stbo{u_1}{u_2}\in \mathrm{Null} \rbkt{\stbt{I}{B}{A}{I}}$ and note that 
\be\label{eq:lem_feedback}
u_1 = -Bu_2~\text{and}~u_2=-Au_1.
\ee 
We next prove that $u=0$ must hold by contradiction. 
  
Suppose that $u\neq 0$. It follows from condition \eqref{eq:lem_separation_sg} that for all $z_1\in \sg(A)$ and $z_2\in \sg^\dagger(B)$, we have $z_1\neq z_2$. Equivalently, we have $\abs{z_1-z_2}>0$. Without loss of generality, assume that $z_1$ and $z_2$ are both in the complex upper half-plane, i.e., $\angle z_1, \angle z_2\in \interval{0}{\pi}$. Then by definition \eqref{eq: matrix_sg}, we have 
  \be  \label{eq:lem_separation}
  \hspace{-2mm} \abs{\gamma(u_1, Au_1) e^{j\theta(u_1, Au_1)} - \gamma(Bu_2, u_2) e^{j\theta(Bu_2, u_2)} }>0. 
  \ee 
  Substituting \eqref{eq:lem_feedback} into \eqref{eq:lem_separation} gives that 
  \begin{align}
  0<&~\left| \gamma(-Bu_2, Au_1) e^{j\theta(-Bu_2, Au_1)} \right. \notag\\
  &\hspace{15mm} \left.- \gamma(Bu_2, -Au_1) e^{j\theta(Bu_2, -Au_1)} \right|\notag\\
  =&~\left| \gamma(Bu_2, Au_1) e^{j\theta(-Bu_2, Au_1)} \right. \notag\\
  & \hspace{15mm}\left.- \gamma(Bu_2, Au_1) e^{j\theta(-Bu_2, Au_1)} \right|, \label{eq:lem_contradiction}
  \end{align}
where the last equality is due to the following identities:
\begin{align*}
  \gamma(x, y)&=\gamma(x, -y)=\gamma(-x, y)\\
  \theta(-x, y)&=\theta(x, -y)
\end{align*}
for all $x, y\in \cm$.
However, the right-hand side of \eqref{eq:lem_contradiction} is obviously equal to zero, which causes a contradiction in terms of \eqref{eq:lem_separation}. Consequently, it holds that $u=0$. This completes the proof.  \hspace*{\fill}~\QED

\emph{Proof of Theorem~\ref{thm:lti_instability}}:~ 
Denote the dominance of $\goflti$ by $p$. First, since there is no unstable pole-zero cancellation in $\goflti$, it follows from Lemma~\ref{lem: dominance_feedback} that showing the $p$-dominance of $\goflti$ is equivalent to showing the $p$-dominance of $(I-PC)^{-1}$. In addition, the cascaded system $PC$ is $(p_1+p_2)$-dominant and so is $I-PC$.
  
Next, we prove that $p=p_1+p_2$ must hold under \eqref{eq:large_sg_frequency} by contradiction. Suppose that $p\neq p_1+p_2$. Denote the difference by $k\coloneqq p-(p_1+p_2)\neq 0$. Note that $PC$ is $(p_1+p_2)$-dominant and $(I-PC)^{-1}$ is $p$-dominant. Since $P, C\in \rlinf^\mm$, we can construct a Nyquist contour $\Gamma=j\mathbb{R}\cup \{\infty\}$. By Cauchy's argument principle,  when $s$ travels up the contour $\Gamma$, the scalar-valued function $\phi(s)\coloneqq \mathrm{det}(I-P(s)C(s))$ encircles the origin $k$ times clockwise  if $k>0$ or counterclockwise if $k<0$. Equivalently,  a $k$ number of encirclement of the critical point $1+j0$ is made by the eigenloci $\lambda_i(P(s)C(s))$ for all $i=1, 2, \ldots, m$ and $s\in \Gamma$.
  
On the one hand, by continuity of the poles of the closed-loop system and the fact that $P(s)C(s)$ has no pole on $s\in \Gamma$, the zeros of $\mathrm{det}(I-\tau P(s)C(s))$ are continuous both in $\tau\in \interval[open left]{0}{1}$ and $s\in \Gamma$. It follows from the nonzero encirclement of the origin that there exists some $\omega_0\in \interval{-\infty}{\infty}$ and some  $\tau_0\in \interval[open left]{0}{1}$ such that $\mathrm{det}(I-\tau_0 P(\jw_0)C(\jw_0))$ exactly intersects the origin, that is,  
\be\label{eq:thm_LTI_contradiction}
\mathrm{det}(I-\tau_0 P(\jw_0) C(\jw_0))=0.
\ee
 
On the other hand, by the separation condition \eqref{eq:large_sg_frequency}, for all $\omega\in \interval{-\infty}{\infty}$ and $\tau \in\interval[open left]{0}{1}$, for all points $z_1\in \sg(\tau P(\jw))$ and $z_2\in \sg^{\dagger}(C(\jw))$, we have $z_1\neq z_2$.  By invoking Lemma~\ref{lem:matrix}, we have
\bex
\mathrm{det}(I-\tau P(\jw) C(\jw))\neq 0
\eex 
for all $\omega\in \interval{-\infty}{\infty}$ and $\tau \in\interval[open left]{0}{1}$. This obviously contradicts  condition~\eqref{eq:thm_LTI_contradiction}. The established contradiction above indicates that $k=0$ and thus the feedback system $\gof$ is $(p_1+p_2)$-dominant. \hspace*{\fill}~\QED
 
\bibliographystyle{IEEEtran}
\bibliography{mybib}

\end{document}